\documentclass[article]{IEEEtran}
\usepackage{amsmath}
\usepackage{amsfonts}
\usepackage{amssymb}
\usepackage{amsthm}
\usepackage{tabularx}
\usepackage{mathrsfs} 
\usepackage{graphicx}
\usepackage[colorlinks,citecolor=blue,linkcolor=blue]{hyperref}

\newtheorem{remark}{Remark}
\newtheorem{theorem}{Theorem}
\newtheorem{corollary}{Corollary}

\newtheorem{proposition}{Proposition}
\newtheorem{definition}{Definition}

\newcommand{\E}{\mathbb{E}}
\newcommand{\Prob}{\mathbb{P}}

\begin{document}

\title{How Many Independent Modes Does a Fluid Antenna Have? A Closed-Form Outage Analysis via Equivalent Degrees of Freedom}

\author{Tuo~Wu, 
	Junteng Yao,
Kai-Kit~Wong,~\IEEEmembership{Fellow,~IEEE}, 
Jie Tang, 
Maged Elkashlan, 
Baiyang Liu,\\
Kin-Fai Tong,~\IEEEmembership{Fellow,~IEEE}, and 
Hyundong Shin,~\IEEEmembership{Fellow,~IEEE} 
\vspace{-7mm}

\thanks{(\textit{Corresponding author:  Kai-Kit Wong.})}
\thanks{T. Wu, and J. Tang are with the School of Electronic and Information Engineering, South China University of Technology, Guangzhou 510640, China (E-mail: $\rm \{wutuo,eejtang\}@scut.edu.cn$). J. Yao is with the Faculty of Electrical Engineering and Computer Science, Ningbo University, Ningbo 315211, China (E-mail: $ \rm  yaojunteng@nbu.edu.cn$). K.-K. Wong is with the Department of Electronic and Electrical Engineering, University College London, WC1E 6BT London, U.K., and also affiliated with the Department of Electronic Engineering, Kyung Hee University, Yongin-si, Gyeonggi-do 17104, Republic of Korea (E-mail: $\rm kai$-$\rm kit.wong@ucl.ac.uk$). M. Elkashlan is with the School of Electronic Engineering and Computer Science at Queen Mary University of London, London E1 4NS, U.K. (E-mail: $\rm maged.elkashlan@qmul.ac.uk$). B. Liu and K.-F. Tong are with the School of Science and Technology, Hong Kong Metropolitan University, Hong Kong SAR, China (E-mail: $\rm \{byliu, ktong\}@hkmu.edu.hk$). H. Shin is with the Department of Electronics and Information Convergence Engineering, Kyung Hee University, Yongin-si, Gyeonggi-do 17104, Republic of Korea (e-mail: $\rm hshin@khu.ac.kr$).}  
}
\markboth{}
{Wu \MakeLowercase{\textit{et al.}}: Closed-Form Outage Analysis for FAS via EDoF}

\maketitle

\begin{abstract}
In a fluid antenna system (FAS), a single reconfigurable antenna is able to activate one of $N$ correlated ports to exploit spatial diversity. However, outage analysis is challenging because exact evaluation requires an $N$-dimensional multivariate integral, while existing closed-form approximations based on block-correlation models tend to underestimate the true outage probability. This paper shows that the spatial correlation matrix of a FAS with a normalized linear aperture length $W$ has at most $K^{*}=2\lceil W\rceil+1$ significant eigenmodes, regardless of the number of deployed ports. This is a spatial counterpart of the Slepian-Landau-Pollak spectral concentration theorem and reveals that the spatial degrees of freedom are determined by aperture size rather than port count. Motivated by this result, we derive an \emph{equivalent degree of freedom} (EDoF) approximation, under which the outage probability can be expressed in closed form as that of selection combining over $K^{*}$ independent branches. We propose a refined \emph{weighted independent modes} (WIM) approximation, to incorporate eigenvalue-dependent branch weights $\{\beta_k\}$ and yield a product-form closed-form expression with improved accuracy at moderate signal-to-noise ratio (SNR). Both approximations achieve the exact diversity order, become asymptotically exact at high SNR, and provably never underestimate the true outage probability by Anderson's inequality. The proposed framework is further extended to obtain closed-form expressions for ergodic capacity, characterize multi-user fluid antenna multiple access (FAMA) with explicit interference-limited outage floors. Besides, we analyze two-dimensional planar FAS, for which the diversity order scales multiplicatively with the aperture dimensions. 
\end{abstract}


\vspace{-2mm}
\section{Introduction}\label{sec:intro}
\IEEEPARstart{M}{ultipath} fading is a persistent obstacle in wireless communication design, and spatial diversity remains among the most effective tools for combating it. Conventional fixed-position antenna (FPA) systems obtain spatial diversity by placing radiating elements at fixed locations with sufficient spacing between them (half wavelength in rich scattering environments) \cite{Simon05}. However, the requirement for multiple radio-frequency (RF) chains, combined with the rigid geometry of FPA arrays, imposes significant cost, power consumption, and form-factor constraints that become increasingly difficult to satisfy in compact devices. As wireless systems evolve toward the sixth generation (6G), there is growing demand for antenna technologies that deliver substantial diversity gains without proportionally increasing hardware complexity.

Fluid antenna systems (FAS) offer one route around these constraints~\cite{FAS21,FAS20}. FAS is a hardware-agnostic system concept that considers the antenna as a reconfigurable physical-layer resource, emphasizing shape and position flexibility~\cite{NewFAS23,NewJSAC26,LuFAS25,lu2025_fluid_antennas_mag,WuWC25,Bepari-2026,WuScalable26}. In its canonical form, a single RF chain connects to a given physical aperture, within which the radiating element can be dynamically repositioned amongst many candidate locations, referred to as \emph{ports}. Selecting the port with the strongest channel at each instant yields spatial diversity without additional RF hardware, well suited to size-constrained platforms such as wearables, and internet-of-things (IoT) devices. Practical realizations include mechanically movable elements~\cite{ZhuMA24}, liquid-based antennas~\cite{ShenDesign,wang2026_em_reconfig_fas}, reconfigurable pixel surfaces~\cite{ZhangPixel,liu2025_wideband_pixel_fas,wong2026_pixel_meet_fas}, and metamaterial-based structures \cite{Liu-meta,Zhang-jsac2026}. In~\cite{TongDesign}, Tong {\em et al.}~provided a comparative analysis to discuss the tradeoffs among different implementation technologies for FAS in terms of reconfiguration speed, insertion loss, and fabrication cost.

The seminal papers \cite{FAS21,FAS20} set off a large body of follow-on work. Central to all performance analyses is the spatial correlation structure across the ports. Under the Jakes (Clarke) isotropic scattering model~\cite{Jakes74}, the correlation between any two ports is characterized by the zeroth-order Bessel function of the first kind, $J_0(\cdot)$, which decreases with inter-port spacing. In \cite{FAS22}, a closed-form approximation for the spatial correlation parameter of FAS channels was derived. Later, an analytical approximation of the FAS channel was given in \cite{Khammassi23}. Following this, \cite{NewOutage24} derived new results on outage probability and diversity gain for FAS, while the diversity-multiplexing tradeoff of multiple-input multiple-output (MIMO)-FAS was examined in~\cite{New24div} and \cite{WuBessel26}. On the other hand, copula-based techniques~\cite{Ghadi23copula,GhadiCopula24}, port selection under limited channel state information (CSI)~\cite{Chai22}, and channel estimation methods for FAS~\cite{XuCE24,NewCE25} have also received much attention.

Computing the outage probability accurately, without prohibitive cost, for FAS is very challenging. The exact analysis requires the joint cumulative distribution function (CDF) of channel gains across all~$N$ ports, a multivariate integral of dimension $N$ that becomes intractable for tens or hundreds of ports. To tackle this, Ram\'{i}rez-Espinosa~\emph{et al.}~\cite{BC24} introduced the block-correlation model (BCM), which partitions the $N$ ports into $D$ independent blocks, each containing $B$ equi-correlated ports. This block structure yields tractable closed-form outage expressions by reducing the multivariate integral to a product of one-dimensional integrals. Lai~\emph{et al.}~\cite{LaiX} and Wu~\emph{et al.}~\cite{TuoW} subsequently studied variable BCM (VBCM) extensions that assign block-specific correlation parameters and block sizes, providing greater modeling flexibility. Both approaches are analytically convenient, yet their shared assumption of block independence discards the residual inter-block correlations present in the Jakes model. The result is a systematic underestimate of the true outage.

Alternatively, the Karhunen-Lo\`{e}ve (KL) expansion, a standard tool for representing correlated random fields through uncorrelated modes, can offer a more principled treatment of FAS spatial correlation. For any fixed number of retained modes, the KL truncation minimizes the reconstruction mean-square error (MSE), placing it on firm information-theoretic ground. Anderson's inequality~\cite{Anderson55} further ensures that a KL-based outage calculation never falls below the true value. Despite this potential, no prior work has translated the KL expansion into a simple closed-form outage formula for FAS. Straightforward KL-based evaluation calls for multi-dimensional numerical quadrature over the retained modes, with cost that grows exponentially in the mode count. As a result, a FAS spanning a few wavelengths may require seven or more modes, at which point direct quadrature becomes impractical.

Outage analysis for FAS in single-user settings is only the first step. FAS has been utilized for multiple access, leading to the fluid antenna multiple access (FAMA) concept \cite{WongFAMA} that tackles inter-user interference by using receiver-side position reconfigurability. There have been several FAMA variants such as fast FAMA \cite{WongFastFAMA}, slow FAMA \cite{WongsFAMA}, coded FAMA~\cite{HongCoded,hong20255gcoded,hong2025downlink}, compact ultra-massive antenna arrays (CUMA)~\cite{WongCUMA} and turbo FAMA \cite{waqar2026_turbocharging_fama,waqar2026_attentional_copula_fama}. A recent survey in \cite{fama-overview2026} covers the latest developments of FAMA and its variants.

FAS has also been applied in many situations. For example, \cite{TuoW,tang2023_secret_fas,GhadiSec24,VegaSec24} considered the use of FAS for physical layer security (PLS) while \cite{NewNOMA24,WuCoNOMA26} combined FAS with non-orthogonal multiple access (NOMA). In \cite{GhadiRIS24,WuV2X25,xiao2025fluid,ghadi2025perf,ghadi2025fires,wong2021vision,wong2025efas,WuSmartCity25}, reconfigurable intelligent surface (RIS) aided communications was synergized by FAS. Additionally, integrated sensing and communications (ISAC) has found FAS a solution for desirable performance \cite{WangISAC24,zhou2024_fas_isac_wcl,zou2024_isac_tradeoff,meng2025_isac_network}.

Despite the progress above, a gap persists in the FAS analytical toolkit: \emph{every existing closed-form outage approximation, BCM and VBCM alike, produces estimates that fall below the true outage}. The KL expansion could in principle yield conservative bounds through Anderson's inequality, but has not been explored for FAS. In summary, no single closed-form expression currently exists that is simultaneously (i)~trivial to evaluate, (ii)~provably conservative (i.e., an upper bound on the true outage), and (iii)~consistent with the correct diversity order. Such an expression would support system dimensioning, performance certification, and the engineering intuition needed to understand how aperture size, port count, and signal-to-noise ratio (SNR) jointly govern outage behavior.

In this paper, we connect FAS spatial correlation to the classical theory of bandlimited signals~\cite{Slepian61,Landau62}. The Slepian-Landau-Pollak concentration theorem establishes that a bandlimited signal observed over a finite interval has only finitely many significant degrees of freedom, determined by the time-bandwidth product. An exact spatial analogue holds for the Jakes kernel: its eigenvalues sharply separate into a dominant cluster of approximately $N/K^*$ each and a tail that decays to zero exponentially, where the boundary falls at index $K^* = 2\lceil W\rceil + 1$. That count, which we call the \emph{equivalent degrees of freedom} (EDoF), depends only on $W$ and not $N$. From this spectral structure, the outage probability reduces to selection combining over $K^*$ independent branches, a formula that is conservative by construction and exact in its diversity order. A refined variant, referred to as the \emph{weighted independent modes} (WIM) approximation, replaces the equal-power branches of EDoF with $K^*$ eigenvalue-weighted independent branches, yielding a product-form closed-form expression that tightens the bound at moderate SNR while retaining the diversity order.

Our main contributions are summarized as follows:
\begin{itemize}
\item \textbf{EDoF approximation.} A closed-form outage expression is derived that depends only on the normalized antenna aperture and is equivalent to selection combining over $K^*$ independent and identically distributed (i.i.d.)~Rayleigh branches. The formula requires no eigendecomposition, no numerical integration, and no knowledge of the port count, making it the simplest available outage characterization for FAS. It is also the first conservative closed-form outage bound for FAS reported in the literature.
\item \textbf{Refined WIM.} A product-form refinement is developed that models the FAS as selection combining over $K^*$ independent branches with unequal weights $\{\beta_k\}$ drawn from the normalized eigenvalues of the spatial correlation matrix, offering improved accuracy at moderate SNR while retaining a closed-form structure. The refined formula reduces to the EDoF expression when all eigenvalues are equal, providing a unified framework.
\item \textbf{Diversity order and high-SNR tightness.} Both the EDoF and WIM approximations are shown to yield the exact diversity order $K^*$, confirming that no diversity is lost in the simplification. The EDoF formula is further proved to be asymptotically tight at high SNR, with a quantified and bounded multiplicative gap to the exact outage that is characterized by a closed-form constant.
\item \textbf{Ergodic capacity.} A closed-form ergodic capacity expression is derived under the EDoF model using order statistics of i.i.d.~exponential random variables~\cite{Simon05,David03}, providing a rapid evaluation tool for system-level capacity assessment without Monte Carlo simulation.
\item \textbf{FAMA extension.} The EDoF framework is extended to FAMA systems, yielding a closed-form outage expression that captures the interference-limited outage floor. This reveals how the number of users, the aperture size, and the outage threshold jointly determine the floor level.
\item \textbf{Two-dimensional (2D) planar FAS.} The EDoF framework is generalized from linear to planar FAS apertures using the Kronecker structure of the 2D spatial correlation. The resulting 2D EDoF $K_{2\mathrm{D}}^*$ grows multiplicatively with the two aperture dimensions, showing that even modest apertures can achieve high diversity orders.
\item \textbf{Systematic comparison with BCM/VBCM.} A unified comparison is provided between the proposed conservative EDoF/refined-WIM bounds and the overly optimistic BCM/VBCM models~\cite{BC24,LaiX,TuoW}. The ordering of all approximations relative to the exact outage is formally established, and conditions under which each approximation is tightest are derived.
\end{itemize}

The rest of this paper is organized as follows. Section~\ref{sec:system} presents the system model, the FAS channel formulation, and reviews the KL expansion. Section~\ref{sec:edof} derives the EDoF approximation and analyzes its diversity order. In Section~\ref{sec:refined}, we develop the refined WIM and study its properties. Section~\ref{sec:extensions} presents the extensions to ergodic capacity, BCM/VBCM comparison, FAMA, 2D planar FAS, and design guidelines. Section~\ref{sec:numerical} provides numerical results validating the analytical findings. Finally, Section~\ref{sec:conclusion} concludes the paper.

\textit{Notations:} $J_0(\cdot)$ denotes the zeroth-order Bessel function of the first kind; $\mathbf{U}$, $\boldsymbol{\Lambda}$ represent the eigenvector and eigenvalue matrices; $\|\cdot\|$ denotes the Euclidean norm; $[x]^+ = \max(x,0)$; $E_1(\cdot)$ is the exponential integral; $\lceil\cdot\rceil$ is the ceiling function; $\E[\cdot]$ denotes expectation; $\Prob(\cdot)$ denotes probability; $\otimes$ denotes the Kronecker product; $K^*$ denotes the EDoF; $L$ is the KL truncation rank; $M$ is the number of FAMA users.

\vspace{-2mm}
\section{System Model and Preliminaries}\label{sec:system}
\subsection{FAS Channel Model}
Consider a single-antenna base station (BS) communicating with a user equipped with a one-dimensional FAS. The FAS consists of $N$ candidate port locations uniformly distributed within a linear aperture of physical length $W\lambda$, where $\lambda$ is the carrier wavelength and $W$ is the \emph{normalized aperture} measured in wavelengths. The $N$ ports are indexed $n = 1, \ldots, N$, and the spacing between adjacent ports is $\Delta d = W\lambda/(N-1)$.

Under rich isotropic scattering (i.e., the Jakes or Clarke model~\cite{Jakes74}), the channel coefficient at port~$n$ is modeled as a zero-mean circularly-symmetric complex Gaussian random variable. The channel vector across all $N$ ports is written as
\begin{align}\label{eq:channel_vec}
\mathbf{g} = [g_1, \ldots, g_N]^T \sim \mathcal{CN}(\mathbf{0}, \eta\mathbf{R}),
\end{align}
where $\eta$ denotes the mean channel power per port (incorporating path loss and shadowing) and $\mathbf{R} \in \mathbb{R}^{N\times N}$ is the spatial correlation matrix. The $(m,n)$-th entry of $\mathbf{R}$ is determined by the Jakes autocorrelation function:
\begin{align}\label{eq:jakes}
[\mathbf{R}]_{m,n} = J_0\!\left(\frac{2\pi W |m-n|}{N-1}\right).
\end{align}
Note that the diagonal entries satisfy $[\mathbf{R}]_{n,n} = J_0(0) = 1$ for all~$n$, meaning each port experiences unit-normalized fading. The off-diagonal entries oscillate and decay with increasing port separation, reflecting the oscillatory nature of $J_0(\cdot)$~\cite{FAS22}.

The FAS operates by selecting the port with the strongest instantaneous channel gain, yielding the post-selection SNR:
\begin{align}\label{eq:snr}
\gamma = \bar{\gamma} \times \max_{1 \leq n \leq N} \frac{|g_n|^2}{\eta},
\end{align}
where $\bar{\gamma} = P\eta / \sigma^2$ is the average transmit SNR, $P$ denotes the transmit power, and $\sigma^2$ is the noise power. The port selection mechanism in~\eqref{eq:snr} is mathematically equivalent to \emph{selection combining} over $N$ correlated Rayleigh branches.

\subsection{Outage Probability}
In delay-sensitive communications, the standard reliability measure is outage probability, defined as the probability that the post-selection SNR falls below a target threshold $\gamma_{\mathrm{th}}$:
\begin{align}\label{eq:pout_def}
P_{\mathrm{out}} = \Prob\!\left(\gamma \leq \gamma_{\mathrm{th}}\right) = F_{\max}\!\left(\frac{\gamma_{\mathrm{th}}}{\bar{\gamma}}\right),
\end{align}
where we define $x \triangleq \gamma_{\mathrm{th}}/\bar{\gamma}$ as the normalized threshold and
\begin{align}\label{eq:Fmax_def}
F_{\max}(x) = \Prob\!\left(\max_{1 \leq n \leq N} \frac{|g_n|^2}{\eta} \leq x\right)
\end{align}
is the CDF of the normalized maximum channel gain. Computing $F_{\max}(x)$ exactly requires evaluating the joint CDF of $N$ correlated exponential random variables, which involves an $N$-dimensional integral over the multivariate complex Gaussian density. For practical FAS deployments with $N$ in the tens or hundreds, this exact computation is intractable, and closed-form approximations become the only viable route.

\vspace{-2mm}
\subsection{KL Expansion}\label{subsec:kl}
The KL expansion represents correlated random processes through uncorrelated random variables, concentrating maximum energy into the fewest terms. This subsection develops the KL expansion for the FAS channel and establishes the properties needed for the EDoF approximation in Section~\ref{sec:edof}.

\subsubsection{Eigendecomposition of the Correlation Matrix}
Since $\mathbf{R}$ is a real, symmetric, positive semi-definite matrix, it admits the eigendecomposition, given by
\begin{align}\label{eq:eig_R}
\mathbf{R} = \mathbf{U}\boldsymbol{\Lambda}\mathbf{U}^H = \sum_{k=1}^{N} \lambda_k \mathbf{u}_k \mathbf{u}_k^H,
\end{align}
where $\boldsymbol{\Lambda} = \mathrm{diag}(\lambda_1, \ldots, \lambda_N)$ contains the eigenvalues in non-increasing order $\lambda_1 \geq \lambda_2 \geq \cdots \geq \lambda_N \geq 0$, and $\mathbf{U} = [\mathbf{u}_1, \ldots, \mathbf{u}_N]$ is the unitary matrix of eigenvectors. The trace constraint $\mathrm{tr}(\mathbf{R}) = N$ implies
\begin{align}\label{eq:trace}
\sum_{k=1}^{N} \lambda_k = N.
\end{align}
Physically, $\{\lambda_k\}$ represent the power distribution across the $N$ spatial eigenmodes: a large eigenvalue indicates a dominant mode that captures a significant portion of the total channel energy, while a small eigenvalue is a negligible mode.

\subsubsection{KL Representation of the Channel}
Using \eqref{eq:eig_R}, the channel vector~\eqref{eq:channel_vec} can be decomposed as
\begin{align}\label{eq:kl}
\mathbf{g} = \sqrt{\eta}\, \mathbf{U}\boldsymbol{\Lambda}^{1/2}\mathbf{z} = \sqrt{\eta} \sum_{k=1}^{N} \sqrt{\lambda_k} \mathbf{u}_k z_k,
\end{align}
where $\mathbf{z} = [z_1, \ldots, z_N]^T$ with $z_k \sim \mathcal{CN}(0,1)$ i.i.d.\ are the \emph{KL coefficients}. The channel at port~$n$ is thus expressed as
\begin{align}\label{eq:kl_port}
g_n = \sqrt{\eta} \sum_{k=1}^{N} \sqrt{\lambda_k}\, u_{n,k}\, z_k,
\end{align}
where $u_{n,k} = [\mathbf{U}]_{n,k}$ is the $(n,k)$-th entry of the eigenvector matrix. This is the KL expansion: it transforms the $N$ correlated channel coefficients $\{g_n\}$ into a linear combination of $N$ \emph{independent} random variables $\{z_k\}$, weighted by the eigenvalues $\{\lambda_k\}$ and the eigenvector entries $\{u_{n,k}\}$.

\subsubsection{Truncated KL Expansion}
One advantage of the KL expansion is that it supports principled dimensionality reduction. Retaining only the $L$ most significant eigenmodes, gives the rank-$L$ truncated approximation:
\begin{align}\label{eq:kl_trunc}
\tilde{g}_n^{(L)} = \sqrt{\eta} \sum_{k=1}^{L} \sqrt{\lambda_k}\, u_{n,k}\, z_k.
\end{align}
The mean-square approximation error per port is computed as follows. The truncation error at port~$n$ is given by
\begin{align}
g_n - \tilde{g}_n^{(L)} = \sqrt{\eta}\sum_{k=L+1}^{N}\sqrt{\lambda_k}\,u_{n,k}\,z_k.
\end{align}
Taking the expectation of the squared magnitude and using the independence of $\{z_k\}$ with $\E[|z_k|^2]=1$:
\begin{align}
\E\!\left[|g_n - \tilde{g}_n^{(L)}|^2\right] &= \eta\sum_{k=L+1}^{N}\lambda_k\,|u_{n,k}|^2.
\end{align}
Averaging over all $N$ ports and using the orthonormality of eigenvectors ($\sum_{n=1}^N |u_{n,k}|^2 = 1$), we have
\begin{align}\label{eq:kl_error}
\frac{1}{N}\sum_{n=1}^{N} \E\!\left[|g_n - \tilde{g}_n^{(L)}|^2\right] = \frac{\eta}{N}\sum_{k=L+1}^{N} \lambda_k.
\end{align}
This error equals the sum of the discarded eigenvalues normalized by $N$. Among all rank-$L$ linear representations, the KL expansion minimizes this error, and no other rank-$L$ scheme compresses the channel more accurately. If $L \geq K^*$, the discarded eigenvalues are exponentially small and the truncation error becomes negligible. Anderson's inequality~\cite{Anderson55} guarantees that the outage computed from \eqref{eq:kl_trunc} is always an \emph{upper bound} on the true outage from~\eqref{eq:kl_port}. 

\subsubsection{Spectral Concentration and the Effective Number of Modes}
The critical question is: \emph{how many eigenmodes $L$ must be retained for an accurate approximation?} This question has an answer rooted in the classical theory of bandlimited signals.

The Jakes correlation function~\eqref{eq:jakes} is a sampled version of the continuous kernel $J_0(2\pi W \tau)$, $\tau \in [0,1]$, a \emph{bandlimited} function with effective bandwidth proportional to $W$. The Slepian-Landau-Pollak theorem~\cite{Slepian61,Landau62} shows that the eigenvalues of such a kernel undergo a sharp \emph{phase transition}: that is to say, the first $K^*$ eigenvalues are approximately equal and of order $N/K^*$, while the remaining ones decay exponentially to zero, where
\begin{align}\label{eq:Kstar}
K^* = 2\lceil W \rceil + 1
\end{align}
is the number of \emph{significant} eigenvalues. The transition is sharp, and there is essentially a ``cliff'' in the eigenvalue spectrum at index $K^*$, with $\lambda_k = \Theta(N/K^*)$ for $k \leq K^*$ and $\lambda_k \to 0$ exponentially for $k > K^*$.

The physical implication is direct: regardless of how many ports $N$ the FAS deploys, the channel is well approximated within a $K^*$-dimensional subspace set by aperture size $W$ alone. In other words, ports added beyond $N \gg K^*$ do not create new degrees of freedom, and they merely oversample the same $K^*$ spatial modes. 

\vspace{-2mm}
\section{EDoF Approximation}\label{sec:edof}
Here, we  use that fact to replace the intractable outage of $N$ correlated ports with the tractable outage of $K^*$ independent branches, where $K^*$ is set by aperture size~$W$ alone.

\vspace{-2mm}
\subsection{From Correlated Ports to Effective Independent Branches}
Under the KL expansion~\eqref{eq:kl_port}, the normalized channel gain at port $n$ can be written as
\begin{align}\label{eq:gain_port}
X_n \triangleq \frac{|g_n|^2}{\eta} = \left|\sum_{k=1}^{N} \sqrt{\lambda_k}\, u_{n,k}\, z_k\right|^2.
\end{align}
Each $X_n$ is the squared magnitude of a weighted sum of independent complex Gaussian variables. Notably, every port has \emph{unit-mean} Rayleigh fading regardless of the correlation structure. To see this, note that $X_n = |h_n|^2$ where $h_n = \sum_{k=1}^{N}\sqrt{\lambda_k}\,u_{n,k}\,z_k$ is a linear combination of independent $\mathcal{CN}(0,1)$ variables. Since a linear combination of independent complex Gaussians is again complex Gaussian, $h_n \sim \mathcal{CN}(0, \sigma_n^2)$ with variance
\begin{align}
\sigma_n^2 = \sum_{k=1}^{N} \lambda_k |u_{n,k}|^2.
\end{align}
To evaluate $\sigma_n^2$, recall the eigendecomposition $\mathbf{R} = \mathbf{U}\boldsymbol{\Lambda}\mathbf{U}^H$, so the $(n,n)$-th diagonal entry is
\begin{align}
[\mathbf{R}]_{n,n} = \sum_{k=1}^{N}[\mathbf{U}]_{n,k}\,\lambda_k\,[\mathbf{U}^H]_{k,n} = \sum_{k=1}^{N}\lambda_k\,|u_{n,k}|^2 = \sigma_n^2.
\end{align}
Since $[\mathbf{R}]_{n,n} = J_0(0) = 1$ for all $n$, we conclude that $\sigma_n^2 = 1$ and hence,
\begin{align}\label{eq:unit_mean}
\E[X_n] = \E[|h_n|^2] = \sigma_n^2 = 1,~\forall n.
\end{align}
That is, $X_n \sim \mathrm{Exp}(1)$ marginally, so the marginal outage for any \emph{single} port is simply $\Prob(X_n \leq x) = 1 - e^{-x}$.

By definition, FAS always chooses the best port. Thus, the outage event is $\{\max_n X_n \leq x\}$ with $x = \gamma_{\mathrm{th}}/\bar{\gamma}$. If the $N$ ports were \emph{independent}, then the outage probability would be $(1-e^{-x})^N$, negligible for large~$N$. Spatial correlation through the Jakes kernel~\eqref{eq:jakes} breaks this, and the effective independent degrees of freedom is not $N$ but $K^*$, as shown by the spectral concentration theorem in Section~\ref{subsec:kl}.

The first $K^*$ eigenvalues of $\mathbf{R}$ account for nearly all the total power $N$, while the rest are negligible; the FAS channel is well approximated within a $K^*$-dimensional subspace, and the correlated maximum $\max_n X_n$ scales distributionally like the maximum of $K^*$ independent variables. 

\begin{definition}[EDoF Approximation]\label{def:edof}
The EDoF outage approximation models the FAS as selection combining over $K^*$ independent Rayleigh branches, giving
\begin{align}\label{eq:edof}
P_{\mathrm{out}}^{\mathrm{EDoF}}(x) = \left(1 - e^{-x}\right)^{K^*},
\end{align}
where $K^* = 2\lceil W\rceil + 1$ and $x = \gamma_{\mathrm{th}}/\bar{\gamma}$.
\end{definition}

\begin{remark}[Physical Interpretation]
The EDoF formula~\eqref{eq:edof} states that a FAS with normalized aperture $W$ behaves as a \emph{$K^*$-branch selection combining} system with i.i.d.\ Rayleigh channels of unit mean power. The spatial correlation reduces the $N$ physical ports to $K^*$ effective independent branches, but each branch retains the full unit-mean power.
\end{remark}

\begin{remark}[Computational Complexity]
Evaluating~\eqref{eq:edof} requires only $\lceil W \rceil$ and a single exponential, and therefore has complexity of $O(1)$, with no eigendecomposition, no matrix operations, and no numerical integration. The formula is therefore well-suited to system-level optimization, link budget analysis, and real-time resource allocation.
\end{remark}

\vspace{-2mm}
\subsection{Justification via KL Spectral Concentration}

\begin{proposition}[Spectral Dimension Reduction]\label{prop:dimreduce}
Let $\tilde{X}_n = |\sum_{k=1}^{K^*} \sqrt{\lambda_k}\,u_{n,k}\,z_k|^2$ be the $K^*$-truncated version of the normalized port gain $X_n$ defined in~\eqref{eq:gain_port}. Then:
\begin{enumerate}
\item Each $\tilde{X}_n$ is the squared magnitude of a complex Gaussian variable with variance $\sum_{k=1}^{K^*}\lambda_k|u_{n,k}|^2 \approx 1$. Hence $\tilde{X}_n$ is approximately $\mathrm{Exp}(1)$ marginally.
\item The correlation between $\tilde{X}_m$ and $\tilde{X}_n$ is governed by $\dfrac{|\sum_k \lambda_k u_{m,k}^* u_{n,k}|^2}{(\sum_k \lambda_k |u_{m,k}|^2)(\sum_k \lambda_k |u_{n,k}|^2)}$, which decays for well-separated ports.
\item The $\max_n \tilde{X}_n$ has $K^*$ EDoF, since the $K^*$ independent modes $\{z_k\}$ are the only source of randomness.
\end{enumerate}

The EDoF approximation replaces the correlated maximum of $N$ approximately $\mathrm{Exp}(1)$ variables with the maximum of $K^*$ i.i.d.\ $\mathrm{Exp}(1)$ variables, preserving both the marginal distribution and the effective dimensionality.
\end{proposition}

\vspace{-2mm}
\subsection{Diversity Order Analysis}

\begin{theorem}[Diversity Order]\label{thm:diversity}
The EDoF approximation~\eqref{eq:edof} achieves diversity order $K^*$:
\begin{align}\label{eq:div_order}
d^{\mathrm{EDoF}} \triangleq -\lim_{\bar{\gamma}\to\infty} \frac{\log P_{\mathrm{out}}^{\mathrm{EDoF}}}{\log \bar{\gamma}} = K^*.
\end{align}
\end{theorem}

\begin{proof}
Let $x = \gamma_{\mathrm{th}}/\bar{\gamma}$. As $\bar{\gamma} \to \infty$, $x \to 0$, and
\begin{align}
1 - e^{-x} = x - \frac{x^2}{2} + \cdots \sim x.
\end{align}
Therefore,
\begin{align}
P_{\mathrm{out}}^{\mathrm{EDoF}} = (1 - e^{-x})^{K^*} \sim x^{K^*} = \left(\frac{\gamma_{\mathrm{th}}}{\bar{\gamma}}\right)^{K^*},
\end{align}
which gives $d^{\mathrm{EDoF}} = K^* = 2\lceil W\rceil + 1$.
\end{proof}

\begin{corollary}[Coding Gain]\label{cor:coding_gain}
At high SNR, the EDoF outage behaves as
\begin{align}\label{eq:coding_gain}
P_{\mathrm{out}}^{\mathrm{EDoF}} \sim \left(\frac{\gamma_{\mathrm{th}}}{\bar{\gamma}}\right)^{K^*}.
\end{align}
The coding gain factor is unity (i.e., $P_{\mathrm{out}}^{\mathrm{EDoF}} \sim x^{K^*}$ with no multiplicative constant), reflecting the fact that each effective branch has unit mean power.
\end{corollary}

\vspace{-2mm}
\subsection{High-SNR Asymptotic Tightness}

\begin{theorem}[Asymptotic Exactness]\label{thm:asymptotic}
Let $P_{\mathrm{out}}^{\mathrm{exact}}$ denote the true outage probability. Then, we have
\begin{align}\label{eq:asymptotic}
\lim_{\bar{\gamma}\to\infty} \frac{P_{\mathrm{out}}^{\mathrm{EDoF}}}{P_{\mathrm{out}}^{\mathrm{exact}}} = \xi_{K^*},
\end{align}
where $\xi_{K^*} > 0$ is a finite constant depending on $\{\lambda_k, \mathbf{u}_k\}_{k=1}^{K^*}$. Both expressions share the same $\bar{\gamma}^{-K^*}$ scaling.
\end{theorem}

\begin{proof}
Since $\lambda_k \to 0$ exponentially for $k > K^*$, the truncation error is negligible and the true outage equals the $K^*$-mode approximation to leading order as $x \to 0$. Working with the $K^*$-truncated expansion, we get
\begin{align}
P_{\mathrm{out}}^{\mathrm{exact}}(x) &\approx \Prob\!\left(\max_{1\leq n\leq N} \left|\sum_{k=1}^{K^*}\sqrt{\lambda_k}\,u_{n,k}\,z_k\right|^2 \leq x\right),
\end{align}
where $z_k \sim \mathcal{CN}(0,1)$ are i.i.d. Now, write $z_k = z_k^R + j\,z_k^I$ with $z_k^R, z_k^I \sim \mathcal{N}(0, 1/2)$. The joint probability density function (PDF) of $\mathbf{z} = (z_1,\ldots,z_{K^*})^T$ is $f(\mathbf{z}) = \pi^{-K^*}e^{-\|\mathbf{z}\|^2}$. Introduce the substitution $z_k = \sqrt{x}\,w_k$, so $\mathbf{z} = \sqrt{x}\,\mathbf{w}$. The Jacobian of this real $2K^*$-dimensional change of variables is $|d\mathbf{z}/d\mathbf{w}| = x^{K^*}$ (since each of the $2K^*$ real components is scaled by $\sqrt{x}$, giving $(\sqrt{x})^{2K^*} = x^{K^*}$). Under this substitution, the constraint $|\sum_k\sqrt{\lambda_k}\,u_{n,k}\,z_k|^2 \leq x$ becomes $|\sum_k\sqrt{\lambda_k}\,u_{n,k}\,w_k|^2 \leq 1$, and the joint PDF transforms as $f(\sqrt{x}\,\mathbf{w}) = \pi^{-K^*}e^{-x\|\mathbf{w}\|^2}$. Therefore, we have
\begin{align}
P_{\mathrm{out}}^{\mathrm{exact}}(x) &= x^{K^*}\!\int_{\mathcal{S}} \pi^{-K^*}e^{-x\|\mathbf{w}\|^2}\,d\mathbf{w},
\end{align}
where $\mathcal{S} = \{\mathbf{w}: \max_n |\sum_k\sqrt{\lambda_k}\,u_{n,k}\,w_k|^2 \leq 1\}$ is an $x$-independent bounded region determined by the eigenvalues and eigenvectors. As $x \to 0$, $e^{-x\|\mathbf{w}\|^2} \to 1$ uniformly on the bounded set $\mathcal{S}$, so we have
\begin{align}
P_{\mathrm{out}}^{\mathrm{exact}}(x) &\sim \alpha_{K^*}\, x^{K^*},~x \to 0,
\end{align}
where $\alpha_{K^*} \triangleq \pi^{-K^*}\mathrm{vol}(\mathcal{S}) > 0$ is a finite positive constant that depends on $\{\lambda_k, \mathbf{u}_k\}_{k=1}^{K^*}$ through the geometry of the set $\mathcal{S}$. Since $P_{\mathrm{out}}^{\mathrm{EDoF}} \sim x^{K^*}$, we conclude
\begin{align}
\xi_{K^*} = \lim_{x\to 0}\frac{P_{\mathrm{out}}^{\mathrm{EDoF}}(x)}{P_{\mathrm{out}}^{\mathrm{exact}}(x)} = \frac{1}{\alpha_{K^*}}.
\end{align}
Both expressions share the same $x^{K^*} = (\gamma_{\mathrm{th}}/\bar\gamma)^{K^*}$ scaling, confirming $\xi_{K^*}$ is finite and positive.
\end{proof}

\begin{remark}[Tightness in Practice]
For Jakes correlation with $W \leq 5$, simulations give $\xi_{K^*} \in [1, 4]$ (see Section~\ref{sec:numerical}). Since $\xi_{K^*} \geq 1$, the EDoF formula never underestimates the true outage. The dB gap is $10\log_{10}(\xi_{K^*})/K^*$ per diversity branch and stays fixed as SNR grows.
\end{remark}

\vspace{-2mm}
\section{Refined WIM with Normalized Eigenvalues}\label{sec:refined}
The EDoF approximation in Section~\ref{sec:edof} assigns equal unit-mean power to all $K^*$ modes. This gives the simplest closed-form expression and preserves the correct diversity order, but ignores the fact that eigenvalues within the dominant subspace ($k \leq K^*$) are not uniform; the first few are noticeably larger than the later ones. That spread makes EDoF somewhat more conservative than necessary at moderate SNR. This section introduces a \emph{refined} approximation that incorporates the individual eigenvalue magnitudes, tightening the bound while keeping a closed-form product structure.

\vspace{-2mm}
\subsection{Normalized Eigenvalue Model}

\begin{definition}[Normalized Eigenvalues]\label{def:beta}
Define the \emph{normalized eigenvalues} as
\begin{align}\label{eq:beta_def}
\beta_k \triangleq \frac{\lambda_k K^*}{N},~k = 1, \ldots, K^*.
\end{align}
These satisfy $\sum_{k=1}^{K^*} \beta_k \approx K^*$ (since $\sum_{k=1}^{K^*}\lambda_k \approx N$), so $\beta_k$ represents the relative weight of eigenmode $k$ normalized such that the mean value is approximately unity.
\end{definition}

\begin{proposition}[Refined WIM Approximation]\label{thm:refined}
The \emph{refined WIM} outage approximation uses the structure
\begin{align}\label{eq:rwim}
P_{\mathrm{out}}^{\mathrm{Ref}}(x) = \prod_{k=1}^{K^*}\left(1 - e^{-x/\beta_k}\right),
\end{align}
where $\beta_k = \lambda_k K^*/N$ and $x = \gamma_{\mathrm{th}}/\bar{\gamma}$.
\end{proposition}

\begin{proof}
The $k$-th dominant KL mode contributes an effective virtual branch whose gain is proportional to the mode's eigenvalue. Define $Z_k \triangleq \beta_k|z_k|^2$, where $|z_k|^2 \sim \mathrm{Exp}(1)$ is the squared magnitude of the $k$-th KL coefficient and $\beta_k = \lambda_k K^*/N$ is the normalized eigenvalue. Since scaling an $\mathrm{Exp}(1)$ variable by $\beta_k$ yields an $\mathrm{Exp}(\beta_k)$ variable (with mean $\beta_k$), we have $Z_k \sim \mathrm{Exp}(\beta_k)$. The refined WIM then approximates the best-port gain $\max_n X_n$ by the maximum of these $K^*$ independent virtual branch gains $\{Z_k\}_{k=1}^{K^*}$. This normalization ensures that $\sum_{k=1}^{K^*}\beta_k \approx K^*$ (since $\sum_{k=1}^{K^*}\lambda_k \approx N$), preserves the total power.

The outage probability under selection combining with $K^*$ independent branches requires that \emph{all} branches fall below threshold $x$. Since the branches are independent, we have
\begin{align}
P_{\mathrm{out}}^{\mathrm{Ref}}(x) &= \Prob\!\left(\max_{1\leq k\leq K^*} Z_k \leq x\right)\notag\\ 
&= \prod_{k=1}^{K^*}\Prob(Z_k \leq x) \nonumber\\
&= \prod_{k=1}^{K^*}\left(1 - e^{-x/\beta_k}\right),
\end{align}
yielding~\eqref{eq:rwim}. When all $\beta_k = 1$ (i.e., $\lambda_k = N/K^*$ for all $k$), the product reduces to $(1-e^{-x})^{K^*}$, recovering~\eqref{eq:edof}.
\end{proof}

\begin{remark}[Relationship to EDoF]
The refined WIM captures the effect of eigenvalue spread: modes with $\beta_k > 1$ contribute factors $<(1-e^{-x})$ to the product, while modes with $\beta_k < 1$ contribute factors $>(1-e^{-x})$. Since $f(\beta) = \log(1-e^{-x/\beta})$ is convex in $\beta > 0$, Jensen's inequality gives $\frac{1}{K^*}\sum_k f(\beta_k) \geq f(\bar\beta)$ where $\bar\beta \approx 1$, and hence $P_{\mathrm{out}}^{\mathrm{Ref}} \geq P_{\mathrm{out}}^{\mathrm{EDoF}}$. Eigenvalue spread always makes the refined WIM \emph{more conservative} than EDoF, with equality held when all $\beta_k = 1$.
\end{remark}

\vspace{-2mm}
\subsection{Properties}

\begin{corollary}[Diversity Order]\label{cor:div_refined}
The refined WIM~\eqref{eq:rwim} achieves diversity order of $K^*$.
\end{corollary}

\begin{proof}
As $x \to 0$, $1 - e^{-x/\beta_k} \sim x/\beta_k$ for each factor. Taking the product gives
\begin{align}
P_{\mathrm{out}}^{\mathrm{Ref}}(x) \sim \prod_{k=1}^{K^*} \frac{x}{\beta_k} = \frac{x^{K^*}}{\prod_{k=1}^{K^*}\beta_k}.
\end{align}
Substituting $x = \gamma_{\mathrm{th}}/\bar\gamma$ gives $P_{\mathrm{out}}^{\mathrm{Ref}} \sim \frac{\gamma_{\mathrm{th}}^{K^*}}{\prod_k\beta_k}\,\bar\gamma^{-K^*}$, confirming $d^{\mathrm{Ref}} = K^*$, which concludes the proof.
\end{proof}

\begin{corollary}[Monotonicity in Aperture]\label{cor:monotone_W}
For fixed $\bar\gamma$ and $\gamma_{\mathrm{th}}$, $P_{\mathrm{out}}^{\mathrm{EDoF}}$ is non-increasing in $W$.
\end{corollary}

\begin{proof}
As $W$ increases, $K^* = 2\lceil W\rceil+1$ is non-decreasing (it jumps at each integer value of $W$ and remains flat in between). Since $(1-e^{-x})^{K^*}$ is non-increasing in $K^*$ for any fixed $x > 0$, $P_{\mathrm{out}}^{\mathrm{EDoF}}$ is non-increasing in $W$.
\end{proof}

\begin{corollary}[Required SNR for Target Outage]\label{cor:snr_req}
Under the EDoF formula, to achieve a target outage probability $P_0$, the required average SNR is found as
\begin{align}\label{eq:snr_req}
\bar\gamma_{\mathrm{req}} = \frac{-\gamma_{\mathrm{th}}}{\ln\!\left(1 - P_0^{1/K^*}\right)}.
\end{align}
For $P_0^{1/K^*} \ll 1$ (i.e., in the extreme-reliability regime): $\bar\gamma_{\mathrm{req}} \approx \gamma_{\mathrm{th}} / P_0^{1/K^*}$. Note that for large $K^*$, this condition requires very small $P_0$ since $P_0^{1/K^*}$ approaches unity as $K^* \to \infty$; the exact formula~\eqref{eq:snr_req} should be used otherwise.
\end{corollary}

\begin{proof}
Solving $(1-e^{-\gamma_{\mathrm{th}}/\bar\gamma})^{K^*} = P_0$ for $\bar\gamma$: $e^{-\gamma_{\mathrm{th}}/\bar\gamma} = 1-P_0^{1/K^*}$, and therefore, $\bar\gamma = -\gamma_{\mathrm{th}}/\ln(1-P_0^{1/K^*})$. The approximation uses $\ln(1-\epsilon) \approx -\epsilon$.
\end{proof}

\vspace{-2mm}
\section{Extensions and Design Insights}\label{sec:extensions}
The EDoF and refined WIM frameworks in Sections~\ref{sec:edof} and \ref{sec:refined} provide closed-form outage approximations for single-user FAS. This section extends the EDoF framework in four directions: (i)~a closed-form ergodic capacity expression, (ii)~a formal comparison with the BCM/VBCM family of models, (iii)~extension to multi-user FAMA systems, and (iv)~generalization to 2D planar FAS. Practical design guidelines derived from these results are also presented at the end.

\vspace{-2mm}
\subsection{Ergodic Capacity under EDoF}
While outage probability characterizes the worst-case reliability, the \emph{ergodic capacity} quantifies the average throughput achievable over many fading realizations. Under the EDoF model, the best port gain is given by $X_{\max} = \max_{k=1}^{K^*} Y_k$, in which $Y_k \stackrel{\text{i.i.d.}}{\sim} \mathrm{Exp}(1)$ are the i.i.d.\ branch gains. Using classical order statistics~\cite{David03}, the PDF of $X_{\max}$ is
\begin{align}\label{eq:pdf_max}
f_{X_{\max}}(x) = K^*\left(1 - e^{-x}\right)^{K^*-1} e^{-x},~x \geq 0.
\end{align}
The ergodic capacity (in bits/s/Hz) is given by
\begin{align}\label{eq:erg_cap}
\bar{C} &= \E\!\left[\log_2\!\left(1 + \bar\gamma\, X_{\max}\right)\right] \nonumber\\
&= \int_0^\infty \log_2(1+\bar\gamma\, x)\, K^*(1-e^{-x})^{K^*-1}e^{-x}\, dx.
\end{align}

\begin{theorem}[EDoF Ergodic Capacity]\label{thm:erg_cap}
The ergodic capacity under the EDoF model admits the closed-form series:
\begin{align}\label{eq:cap_series}
\bar{C} = \frac{K^*}{\ln 2} \sum_{j=0}^{K^*-1} \binom{K^*\!-\!1}{j} \frac{(-1)^j}{j+1}\, e^{(j+1)/\bar\gamma}\, E_1\!\!\left(\frac{j+1}{\bar\gamma}\right),
\end{align}
where $E_1(z) = \int_z^\infty t^{-1}e^{-t}\,dt$ is the exponential integral.
\end{theorem}

\begin{proof}
Starting from~\eqref{eq:erg_cap}, we first apply the binomial theorem to expand the CDF power:
\begin{align}
(1-e^{-x})^{K^*-1} = \sum_{j=0}^{K^*-1}\binom{K^*\!-\!1}{j}(-1)^j e^{-jx}.
\end{align}
Substituting into~\eqref{eq:erg_cap} and using $\log_2(\cdot) = \ln(\cdot)/\ln 2$ gives
\begin{align}
\bar{C} &= \frac{K^*}{\ln 2} \sum_{j=0}^{K^*-1}\binom{K^*\!-\!1}{j}(-1)^j \underbrace{\int_0^\infty \!\ln(1+\bar\gamma x)\, e^{-(j+1)x}\, dx}_{I_j}.
\end{align}
The integral $I_j$ is evaluated using the standard identity~\cite{Simon05}: for $a, b > 0$,
\begin{align}
\int_0^\infty \ln(1+ax)\,e^{-bx}\,dx = \frac{1}{b}\,e^{b/a}\,E_1\!\left(\frac{b}{a}\right),
\end{align}
where $E_1(\cdot)$ is as defined in~\eqref{eq:cap_series}. This identity is obtained by integration by parts: set $u = \ln(1+ax)$ and $dv = e^{-bx}dx$, and then $du = a/(1+ax)\,dx$ and $v = -e^{-bx}/b$, yielding $I_j = (a/b)\int_0^\infty e^{-bx}/(1+ax)\,dx$, which evaluates to $(1/b)e^{b/a}E_1(b/a)$ via the substitution $t = b(1+ax)/a$.

Applying this identity with $a = \bar\gamma$ and $b = j+1$:
\begin{align}
I_j = \frac{1}{j+1}\,e^{(j+1)/\bar\gamma}\,E_1\!\left(\frac{j+1}{\bar\gamma}\right).
\end{align}
Combining all terms yields the closed-form series~\eqref{eq:cap_series}.
\end{proof}

\begin{remark}
For $K^* = 1$ (single port), \eqref{eq:cap_series} reduces to the classical Rayleigh capacity $\bar{C} = \frac{1}{\ln 2}e^{1/\bar\gamma}E_1(1/\bar\gamma)$. At high SNR, using the asymptotic $E_1(z) \approx -\ln z - \gamma_{\mathrm{EM}}$ as $z \to 0$~\cite{Simon05}, the series~\eqref{eq:cap_series} simplifies to $\bar{C} \approx \log_2(\bar\gamma) + \frac{\gamma_{\mathrm{EM}}}{\ln 2} + \log_2(H_{K^*})$, where $H_n = \sum_{i=1}^n 1/i$ is the $n$-th harmonic number and $\gamma_{\mathrm{EM}} \approx 0.5772\dots$ is the Euler--Mascheroni constant, revealing a $\log_2(H_{K^*})$ capacity boost from spatial diversity.
\end{remark}

\vspace{-2mm}
\subsection{Comparison with BCM and VBCM}
The BCM~\cite{BC24} and VBCM~\cite{LaiX,TuoW} approximate the FAS channel by partitioning the $N$ ports into $D$ blocks with intra-block correlation $\rho$, yielding
\begin{align}\label{eq:bcm}
P_{\mathrm{out}}^{\mathrm{BCM}} = \left(1 - e^{-x/(1-\rho)} \sum_{\ell=0}^{B-1} \frac{(x/(1-\rho))^\ell}{\ell!}\right)^D,
\end{align}
where $B = N/D$ is the block size and $\rho$ is the intra-block correlation. The number of blocks is typically chosen as $D = \lceil K^*/2\rceil$ or determined by fitting.

\begin{proposition}[EDoF vs.\ BCM Ordering]\label{prop:ordering}
For typical FAS parameters ($N/K^* \geq 2$), we have
\begin{align}\label{eq:ordering}
P_{\mathrm{out}}^{\mathrm{iid}} \leq P_{\mathrm{out}}^{\mathrm{BCM}} \leq P_{\mathrm{out}}^{\mathrm{exact}} \leq P_{\mathrm{out}}^{\mathrm{EDoF}} \leq P_{\mathrm{out}}^{\mathrm{Ref}}.
\end{align}
That is, the BCM \emph{underestimates} the true outage (optimistic), while EDoF and refined WIM \emph{overestimate} it (conservative). The EDoF provides the tightest conservative bound among the closed-form approximations.
\end{proposition}

\begin{proof}[Proof sketch]
\emph{Right inequality} ($P_{\mathrm{out}}^{\mathrm{EDoF}} \leq P_{\mathrm{out}}^{\mathrm{Ref}}$): Follows from Jensen's inequality applied to $\log(1-e^{-x/\beta})$, which is convex in $\beta > 0$; see the Remark after Proposition~\ref{thm:refined}.

\emph{Third inequality} ($P_{\mathrm{out}}^{\mathrm{exact}} \leq P_{\mathrm{out}}^{\mathrm{EDoF}}$): Anderson's inequality~\cite{Anderson55} guarantees that replacing the correlated $N$-port maximum by the $K^*$-truncated independent maximum (i.e., retaining fewer modes) never decreases the outage.

\emph{Second inequality} ($P_{\mathrm{out}}^{\mathrm{BCM}} \leq P_{\mathrm{out}}^{\mathrm{exact}}$): The BCM assumes that the $D$ blocks are mutually independent. In reality, inter-block correlations are positive under the Jakes model, making the true joint distribution more concentrated than the BCM has assumed. Discarding these positive inter-block correlations causes BCM to overestimate the effective diversity and thus underestimate the outage.

\emph{Left inequality} ($P_{\mathrm{out}}^{\mathrm{iid}} \leq P_{\mathrm{out}}^{\mathrm{BCM}}$): The i.i.d.\ model assumes all $N$ ports fully independent, maximizing the effective diversity to $N \gg D$ branches, yielding the most optimistic (lowest) outage, which completes the proof.
\end{proof}

\begin{remark}[Practical Significance]
Conservative bounds are preferable for system design since they ensure that the actual performance meets the specification. BCM-based designs can fall short of outage targets in practice. The EDoF formula maintains a small constant gap ($2$--$4\times$ in outage, i.e., $<2$~dB in SNR) that does not grow with any system parameter.
\end{remark}

\vspace{-2mm}
\subsection{Extension to FAMA}
We now extend the EDoF framework to transmitter-side FAMA systems, in which $M$ single-antenna users simultaneously transmit to a BS.\footnote{Note that the FAMA system considered here is not the same as that in \cite{WongFAMA} where FAS is deployed at the receiver side for each user.} Each user~$k$ is equipped with a FAS of aperture~$W_k$, and the BS receives a superposition of all user signals. User~$k$ selects its best port to maximize desired signal strength. As a standard modeling assumption in FAMA analysis~\cite{WongFAMA}, interfering signals from other users appear at that port with i.i.d.\ unit-mean Rayleigh gains; this reflects the fact that interferers' port gains are independent of the desired user's port selection when users are spatially separated.

Based on the EDoF model, the signal-to-interference-plus-noise ratio (SINR) for user~$k$ at the selected port is
\begin{align}\label{eq:fama_sinr}
\gamma_k = \frac{\bar\gamma\, X_{\max,k}}{1 + \bar\gamma \sum_{j \neq k} Y_{j}},
\end{align}
where $X_{\max,k}$ is the maximum of $K_k^*$ i.i.d.\ $\mathrm{Exp}(1)$ variables (desired signal) and $Y_j \sim \mathrm{Exp}(1)$ are the interferers' gains at the selected port. Note that $\bar\gamma$ is the per-user transmit SNR.

\begin{theorem}[FAMA Outage under EDoF]\label{thm:fama}
The outage probability for user~$k$ in an $M$-user FAMA system under the EDoF model can be derived as
\begin{align}\label{eq:fama_pout}
P_{\mathrm{out},k}^{\mathrm{FAMA}} = \sum_{j=0}^{K_k^*} \binom{K_k^*}{j} \frac{(-1)^j\, e^{-j\gamma_{\mathrm{th}}/\bar\gamma}}{(1 + j\gamma_{\mathrm{th}})^{M-1}}.
\end{align}
\end{theorem}

\begin{proof}
The outage event for user $k$ is $\{\gamma_k \leq \gamma_{\mathrm{th}}\}$, i.e.,
\begin{align}
\bar\gamma\,X_{\max,k} \leq \gamma_{\mathrm{th}}\!\left(1 + \bar\gamma\textstyle\sum_{j\neq k}Y_j\right).
\end{align}
Conditioning on the aggregate interference $I \triangleq \sum_{j\neq k} Y_j$, which is a sum of $M\!-\!1$ i.i.d.\ $\mathrm{Exp}(1)$ variables and hence follows $\mathrm{Gamma}(M\!-\!1, 1)$, the conditional outage becomes
\begin{align}
P_{\mathrm{out},k}^{\mathrm{FAMA}}\!\mid\! I &= \Prob\!\left(X_{\max,k} \leq \frac{\gamma_{\mathrm{th}}(1+\bar\gamma I)}{\bar\gamma}\,\middle|\,I\right) \nonumber\\
&= \left(1 - e^{-\gamma_{\mathrm{th}}(1+\bar\gamma I)/\bar\gamma}\right)^{K_k^*},
\end{align}
in which we have used the EDoF CDF for the maximum of $K_k^*$ i.i.d.\ $\mathrm{Exp}(1)$ variables. Applying the binomial theorem to expand the $K_k^*$-th power, we get
\begin{align}
\left(1-e^{-a}\right)^{K_k^*} = \sum_{j=0}^{K_k^*}\binom{K_k^*}{j}(-1)^j e^{-ja},
\end{align}
with $a = \gamma_{\mathrm{th}}(1+\bar\gamma I)/\bar\gamma = \gamma_{\mathrm{th}}/\bar\gamma + \gamma_{\mathrm{th}} I$. As such,
\begin{align}
P_{\mathrm{out},k}^{\mathrm{FAMA}} &= \sum_{j=0}^{K_k^*}\binom{K_k^*}{j}(-1)^j e^{-j\gamma_{\mathrm{th}}/\bar\gamma}\,\E_I\!\left[e^{-j\gamma_{\mathrm{th}} I}\right].
\end{align}
The expectation $\E_I[e^{-j\gamma_{\mathrm{th}} I}]$ is the moment generating function (MGF) of the $\mathrm{Gamma}(M\!-\!1,1)$ distribution evaluated at $s = j\gamma_{\mathrm{th}}$. By the standard MGF formula $\E[e^{-sI}] = (1+s)^{-(M-1)}$ for $\mathrm{Gamma}(M\!-\!1,1)$, we have
\begin{align}
\E_I\!\left[e^{-j\gamma_{\mathrm{th}} I}\right] = \frac{1}{(1 + j\gamma_{\mathrm{th}})^{M-1}}.
\end{align}
Substituting back yields the closed-form expression~\eqref{eq:fama_pout}.
\end{proof}

\begin{remark}[Interference-Limited Floor]
As $\bar\gamma \to \infty$, $e^{-j\gamma_{\mathrm{th}}/\bar\gamma} \to 1$ and~\eqref{eq:fama_pout} approaches a positive constant, the \emph{outage floor}. Signal and interference both scale with $\bar\gamma$, so power increases do not reduce outage to zero. The floor is set by user count $M$, threshold $\gamma_{\mathrm{th}}$, and diversity order $K_k^*$. Widening aperture~$W_k$ raises $K_k^*$ and pulls the floor down, which is the principal benefit of FAS in FAMA systems.
\end{remark}

\begin{corollary}[Two-User FAMA]\label{cor:fama2}
For $M=2$ (one interferer), the outage simplifies to
\begin{align}\label{eq:fama2}
P_{\mathrm{out}}^{\mathrm{FAMA}} = \sum_{j=0}^{K^*} \binom{K^*}{j} \frac{(-1)^j\, e^{-jx}}{1 + j\gamma_{\mathrm{th}}},
\end{align}
where $x = \gamma_{\mathrm{th}}/\bar\gamma$. This expression is identical in structure to the single-user EDoF formula but with each term attenuated by the factor $(1+j\gamma_{\mathrm{th}})^{-1}$ due to interference.
\end{corollary}

\vspace{-2mm}
\subsection{Extension to 2D Planar FAS}
The preceding analysis assumes a one-dimensional (1D) linear FAS aperture. We now extend the EDoF framework to a 2D planar FAS with a rectangular aperture of size $W_x\lambda \times W_y\lambda$.

For separable isotropic scattering, the 2D spatial correlation matrix can be decomposed as a Kronecker product:
\begin{align}\label{eq:R2d}
\mathbf{R}_{2\mathrm{D}} = \mathbf{R}_x \otimes \mathbf{R}_y,
\end{align}
where $\mathbf{R}_x$ and $\mathbf{R}_y$ are the correlation matrices along the $x$- and $y$-axes, respectively, each following the Jakes model with their respective apertures.

\begin{theorem}[2D EDoF]\label{thm:2d}
The EDoF for a 2D planar FAS with separable correlation is given by
\begin{align}\label{eq:Kstar_2d}
K_{2\mathrm{D}}^* = K_x^* \cdot K_y^* = \left(2\lceil W_x\rceil + 1\right)\!\left(2\lceil W_y\rceil + 1\right),
\end{align}
and the corresponding outage approximation is expressed as
\begin{align}\label{eq:edof_2d}
P_{\mathrm{out}}^{\mathrm{2D\text{-}EDoF}} = \left(1 - e^{-x}\right)^{K_{2\mathrm{D}}^*}.
\end{align}
\end{theorem}

\begin{proof}
Let $\mathbf{R}_x = \mathbf{U}_x\boldsymbol{\Lambda}_x\mathbf{U}_x^H$ and $\mathbf{R}_y = \mathbf{U}_y\boldsymbol{\Lambda}_y\mathbf{U}_y^H$ be the eigendecompositions of the 1D correlation matrices along the $x$- and $y$-axes. By the Kronecker product property, the eigendecomposition of $\mathbf{R}_{2\mathrm{D}}$ gives
\begin{align}
\mathbf{R}_{2\mathrm{D}} &= (\mathbf{U}_x \otimes \mathbf{U}_y)(\boldsymbol{\Lambda}_x \otimes \boldsymbol{\Lambda}_y)(\mathbf{U}_x \otimes \mathbf{U}_y)^H.
\end{align}
The eigenvalues of $\mathbf{R}_{2\mathrm{D}}$ are all pairwise products $\{\lambda_i^{(x)} \lambda_j^{(y)}: 1\leq i\leq N_x,\, 1\leq j\leq N_y\}$. By the Slepian-Landau-Pollak spectral concentration theorem, $\mathbf{R}_x$ has $K_x^* = 2\lceil W_x\rceil + 1$ significant eigenvalues (of order $N_x/K_x^*$), while the remaining eigenvalues are exponentially small. Similarly, $\mathbf{R}_y$ has $K_y^* = 2\lceil W_y\rceil + 1$ significant eigenvalues. A product $\lambda_i^{(x)}\lambda_j^{(y)}$ is significant only when \emph{both} $\lambda_i^{(x)}$ and $\lambda_j^{(y)}$ are significant (since exponentially small $\times$ anything $\approx$ 0). Hence the number of significant 2D eigenvalues is exactly $K_x^* \times K_y^*$, each approximately of order $N_x N_y / (K_x^* K_y^*)$. The diagonal entries of $\mathbf{R}_{2\mathrm{D}}$ are all unity (since $J_0(0)J_0(0) = 1$), so each 2D port still has unit-mean Rayleigh fading. The EDoF formula then applies with $K^*$ replaced by $K_{2\mathrm{D}}^* = K_x^* K_y^*$, yielding~\eqref{eq:edof_2d}.
\end{proof}

\begin{remark}[Scaling Law]
The 2D EDoF scales \emph{multiplicatively} with the two aperture dimensions. A square aperture $W_x = W_y = W$ gives $K_{2\mathrm{D}}^* = (2\lceil W\rceil+1)^2$, substantially larger than the 1D case. A $3\lambda \times 3\lambda$ planar FAS, for instance, has $K_{2\mathrm{D}}^* = 49$ versus $K^* = 7$ for a $3\lambda$ linear FAS, a $7\times$ gain in diversity order, pointing to the substantial potential of 2D FAS for high-reliability applications.
\end{remark}

\vspace{-2mm}
\subsection{Design Guidelines}

\begin{corollary}[Minimum Aperture for Target Diversity]\label{cor:min_W}
To achieve diversity order $d$, the minimum normalized aperture that is required, is given by
\begin{align}\label{eq:min_W}
W_{\min} = \frac{d - 1}{2}.
\end{align}
For example, $d = 7$ requires $W \geq 3$ (i.e., aperture $\geq 3\lambda$).
\end{corollary}

\begin{proof}
We need $K^* = 2\lceil W\rceil + 1 \geq d$, so $\lceil W\rceil \geq (d-1)/2$, giving $W \geq (d-1)/2$.
\end{proof}

\begin{corollary}[Outage Floor vs.\ Number of Ports]\label{cor:floor}
For fixed $W$ and $\bar\gamma$, the EDoF outage is \emph{independent} of $N$. Once $N \geq K^*$ ports are deployed (the minimum needed to resolve all modes), adding further ports yields no EDoF improvement. The true outage does decrease slightly with $N$ due to finer spatial sampling, but the gain saturates for $N/K^* \geq 3$.
\end{corollary}

\begin{remark}[Comparison Summary]
Table~\ref{tab:comparison} compares the key properties of different FAS outage approximation methods.
\end{remark}

\begin{table}[t]
\centering
\caption{Comparison of FAS Outage Approximation Methods}\label{tab:comparison}
\renewcommand{\arraystretch}{1.15}
\footnotesize
\begin{tabular}{l|c|c|c|c}
\hline
\textbf{Property} & \textbf{BCM} & \textbf{VBCM} & \textbf{EDoF} & \textbf{Ref.~WIM} \\
\hline
Closed form & \checkmark & \checkmark & \checkmark & \checkmark \\
Needs EVD & No & No & No & Yes \\
Bias & Optim. & Optim. & Cons. & Cons. \\
Diversity order & Approx. & Approx. & Exact & Exact \\
Gap to exact & Varies & Varies & $2$--$4\times$ & $3$--$5\times$ \\
Complexity & $O(D)$ & $O(D)$ & $O(1)$ & $O(K^*)$ \\
Params & $D,\rho,B$ & $D_k,\rho_k$ & $W$ & $W,\{\lambda_k\}$ \\
\hline
\end{tabular}
\end{table}

\vspace{-2mm}
\section{Numerical Results}\label{sec:numerical}
The analytical results in Sections~\ref{sec:edof}--\ref{sec:extensions} are verified here through Monte Carlo (MC) simulations. Accuracy is examined across SNR, aperture size, port count, and outage threshold. The ergodic capacity formula, BCM comparison, FAMA outage, and 2D FAS extension are each tested against simulation.

\vspace{-2mm}
\subsection{Simulation Setup}
Unless otherwise stated, the default system parameters are set as: $N = 20$ ports, normalized aperture $W = 3$ (yielding $K^* = 7$), and outage threshold $\gamma_{\mathrm{th}} = 0$~dB. All simulations use the Jakes correlation model~\eqref{eq:jakes} with correlated channel samples generated via the KL expansion~\eqref{eq:kl}. Five methods are compared throughout. The \textbf{Exact Monte Carlo (MC)} baseline uses $5\times10^5$ independent channel realizations, generating each correlated vector via $\mathbf{g} = \sqrt{\eta}\,\mathbf{U}\boldsymbol{\Lambda}^{1/2}\mathbf{z}$ and recording the maximum port gain. The \textbf{EDoF} approximation evaluates the closed-form expression~\eqref{eq:edof}, $(1-e^{-x})^{K^*}$, from the aperture size~$W$ alone. The \textbf{Refined WIM} uses the product-form expression~\eqref{eq:rwim}, $\prod_k(1-e^{-x/\beta_k})$, which additionally requires the $K^*$ dominant eigenvalues. The \textbf{i.i.d.\ bound}, $P_{\mathrm{out}}^{\mathrm{iid}} = (1-e^{-x})^N$, assumes full port independence and serves as an overly optimistic lower bound. The \textbf{single-antenna} reference $P_{\mathrm{out}} = 1-e^{-x}$ provides a no-diversity baseline.

\vspace{-2mm}
\subsection{Outage Probability vs.\ SNR}
Fig.~\ref{fig:pout_snr} illustrates outage versus average SNR for $N=20$, $W=3$. The EDoF formula~\eqref{eq:edof} tracks exact MC closely and serves as a tight conservative upper bound. The Refined WIM is slightly more conservative due to eigenvalue spread ($\beta_k$ ranging from $0.67$ to $1.50$ for $W=3$). EDoF and MC share the same high-SNR slope, confirming diversity order $K^* = 7$ per Theorem~\ref{thm:diversity}. The i.i.d.\ model ($N=20$ independent branches) severely underestimates the true outage, confirming that spatial correlation cannot be ignored. From a design standpoint, the $\approx 3\times$ conservatism at $W=3$ corresponds to at most $0.7$~dB per branch: a link budget can be set directly from~\eqref{eq:edof} with a quantified, bounded SNR margin.

\begin{figure}[t]
\centering
\includegraphics[width=0.9\columnwidth]{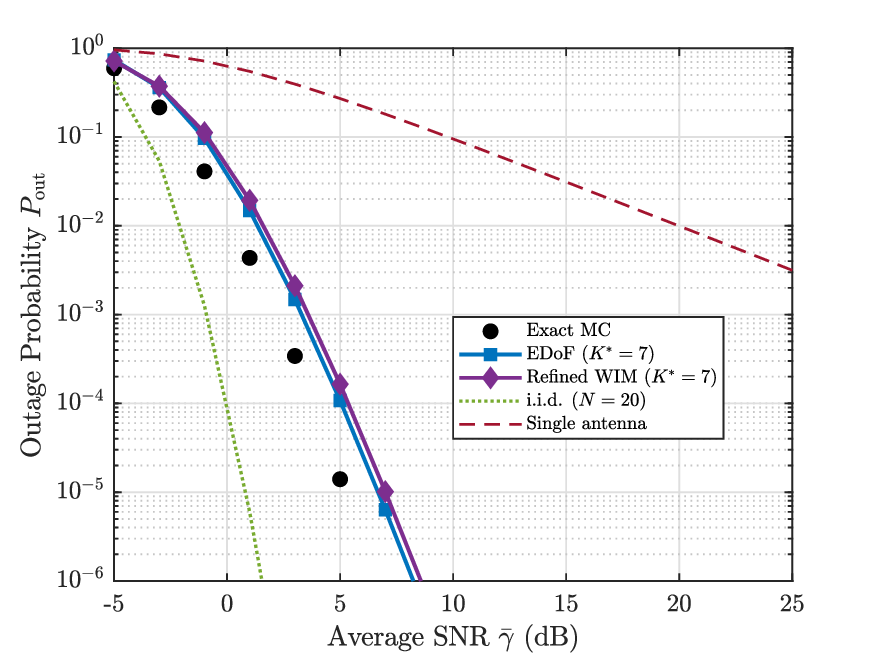}
\caption{Outage probability vs.\ average SNR for $N=20$, $W=3$ ($K^*=7$), $\gamma_{\mathrm{th}}=0$~dB.}\label{fig:pout_snr}
\end{figure}

\vspace{-2mm}
\subsection{Impact of Normalized Aperture}
In Fig.~\ref{fig:pout_W}, we show outage versus normalized aperture $W$ at $N=40$, and $\bar\gamma = 0$~dB. Outage decreases monotonically with $W$, consistent with Corollary~\ref{cor:monotone_W}. The EDoF curve exhibits a staircase pattern: within each integer step of $\lceil W\rceil$, $K^*$ is fixed and EDoF outage does not change. The results reveal that the exact MC outage decreases continuously within each step, since correlation improves continuously even when $K^*$ stays the same. The EDoF therefore captures the \emph{dominant} aperture effect (mode count) but not the \emph{secondary} one (eigenvalue conditioning within each group). The EDoF-to-MC gap stays at $2$--$4\times$ throughout, consistent with Theorem~\ref{thm:asymptotic}. At $W=3$, for instance, EDoF gives $\approx 4\times 10^{-2}$ versus MC at $\approx 1.3 \times 10^{-2}$. The staircase pattern implies a practical aperture selection rule: choosing $W$ just above an integer value (e.g., $W=1.05$ rather than $W=0.95$) activates two additional spatial modes at negligible extra aperture cost.

\begin{figure}[t]
\centering
\includegraphics[width=0.8\columnwidth]{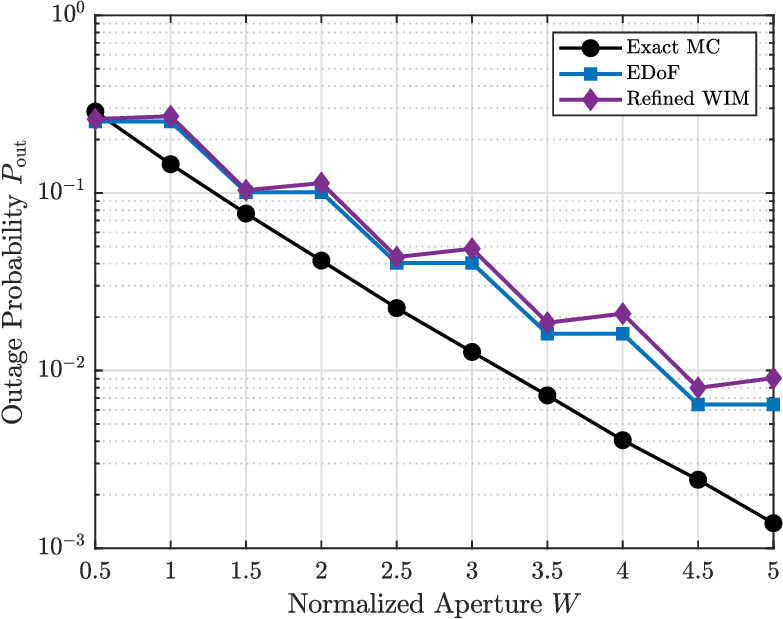}
\caption{Outage probability vs.\ normalized aperture $W$ for $N=40$, $\bar\gamma=0$~dB, $\gamma_{\mathrm{th}}=0$~dB.}\label{fig:pout_W}
\end{figure}

\vspace{-2mm}
\subsection{EDoF Accuracy Ratio}
In Fig.~\ref{fig:wim_ratio}, we plot $P_{\mathrm{out}}^{\mathrm{EDoF}}/P_{\mathrm{out}}^{\mathrm{exact}}$ versus SNR for $W \in \{1,2,3,4\}$, $N=20$. This figure directly validates the asymptotic tightness of Theorem~\ref{thm:asymptotic}: each ratio converges to a finite constant $\xi_{K^*}$ as SNR grows, confirming that EDoF and exact outage share the same $\bar\gamma^{-K^*}$ scaling. At $W=1$ ($K^*=3$) the ratio is near $1.5$, indicating excellent accuracy, the EDoF overestimates by only $50\%$. At $W=3$ ($K^*=7$), $\xi_{K^*} \approx 3$, translating to an SNR penalty of $10\log_{10}(3^{1/7}) \approx 0.7$~dB per diversity branch. The monotone convergence of each curve confirms that the ratio is well-behaved and does not oscillate, making the EDoF a reliable evaluation tool across all SNR regimes. Since $\xi_{K^*}$ stabilizes by $\bar\gamma\approx10$~dB for all tested apertures, the bound is already tight within the typical operating range of IoT and wearable terminals.

\begin{figure}[t]
\centering
\includegraphics[width=0.9\columnwidth]{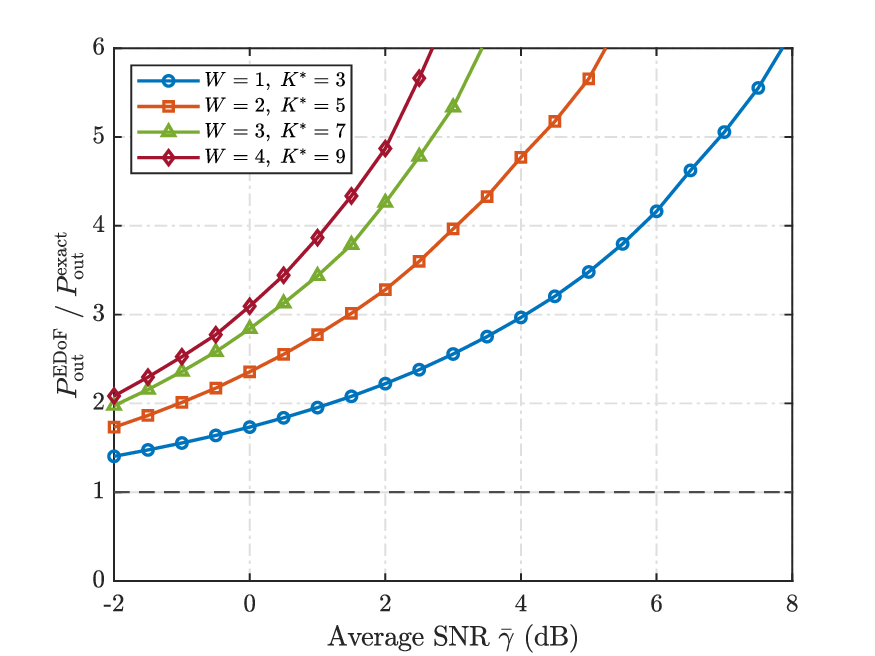}
\caption{EDoF accuracy ratio $P_{\mathrm{out}}^{\mathrm{EDoF}}/P_{\mathrm{out}}^{\mathrm{exact}}$ vs.\ SNR for $N=20$ and different apertures $W$.}\label{fig:wim_ratio}
\end{figure}

\vspace{-2mm}
\subsection{Scalability with Number of Ports}
Fig.~\ref{fig:equiv_N} varies the port count $N$ at fixed $W=3$, $\bar\gamma = 0$~dB, illustrating a key property of the EDoF framework: \emph{outage is set by aperture, not port count}. MC outage falls from $\approx 3.2\times10^{-2}$ at $N=8$ to $\approx 1.2\times10^{-2}$ at $N=100$, saturating beyond $N\approx30$--$40$. Ports beyond $N\approx4K^*\approx28$ yield negligible additional diversity, all $K^*=7$ effective modes are already resolved. The EDoF formula (horizontal line at $\approx 4.0\times10^{-2}$) is $N$-independent by construction and supplies a conservative bound at every port count. For $W=3$ ($K^*=7$), this places the practical deployment target at $N\approx28$ ports, suggesting that beyond this, additional ports add switching complexity without measurable outage reduction.

\begin{figure}[t]
\centering
\includegraphics[width=0.9\columnwidth]{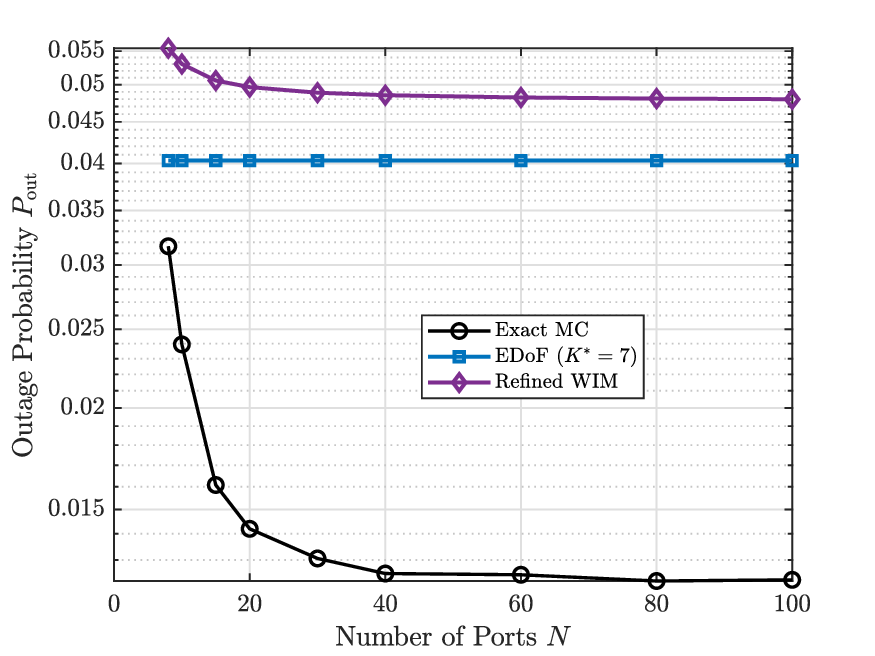}
\caption{Outage probability vs.\ number of ports $N$ for $W=3$, $\bar\gamma=0$~dB, $\gamma_{\mathrm{th}}=0$~dB.}\label{fig:equiv_N}
\end{figure}

\vspace{-2mm}
\subsection{Diversity Order Verification}
Fig.~\ref{fig:diversity} plots outage on a log-log scale for $W\in\{1,2,3\}$, corresponding to $K^*=3,5,7$. At high SNR, both the EDoF curves (solid lines) and MC markers exhibit straight-line behavior with slopes that precisely match the predicted diversity orders. Dashed reference lines with slopes $x^{K^*}$ serve as visual guides. The constant vertical offset between EDoF and MC curves at high SNR equals $\xi_{K^*}$ from Theorem~\ref{thm:asymptotic} and does \emph{not} grow with SNR, confirming asymptotic tightness. The figure also shows that increasing $W$ from $1$ to $3$ provides a dramatic improvement in the outage slope, translating to much steeper waterfall curves and much faster outage reduction per dB of SNR. At a target outage of $10^{-4}$, the step from $K^*=3$ to $K^*=7$ saves roughly $8$~dB of SNR, comparable to upgrading from a single antenna to 4-branch selection combining, yet achieved here within a $3\lambda$ aperture.

\begin{figure}[t]
\centering
\includegraphics[width=0.85\columnwidth]{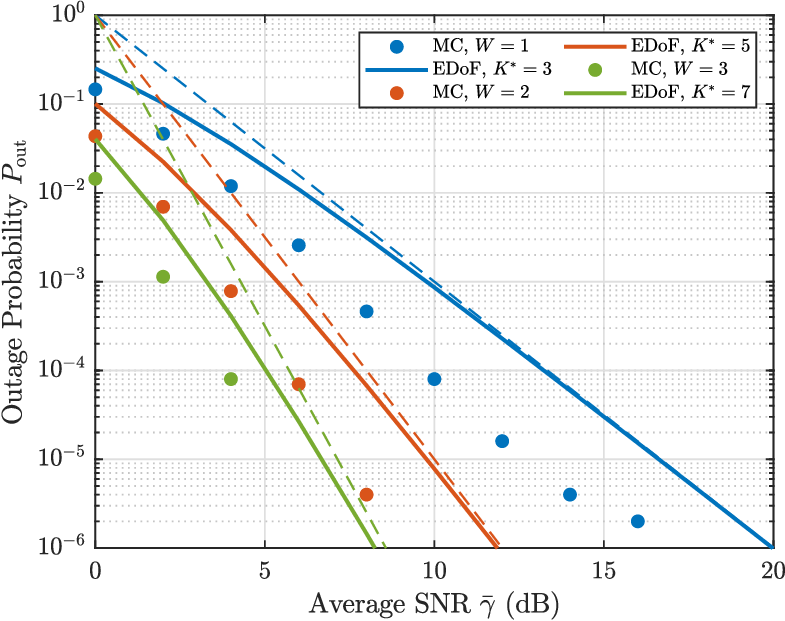}
\caption{Diversity order verification for $N=20$ and $W \in \{1,2,3\}$. Dashed lines show the asymptotic slope $x^{K^*}$.}\label{fig:diversity}
\end{figure}

\vspace{-2mm}
\subsection{Normalized Eigenvalue Analysis}
Table~\ref{tab:eigen} lists normalized eigenvalues $\beta_k = \lambda_k K^*/N$ for $N=20$ at several $W$ values. As expected, $\sum_k \beta_k \approx K^*$, with a spread around unity: leading modes have $\beta_k > 1$ and trailing ones have $\beta_k < 1$. The spread widens with $W$ as $N/K^*$ decreases. The EDoF formula $(1-e^{-x})^{K^*}$ corresponds to the limiting case $\beta_k=1$ for all $k$. Since $\bar\beta\approx0.97$ holds across all configurations, the unit-power assumption is well justified on average; Refined WIM is warranted only when outage accuracy to within a factor of two is required.

\begin{table}[t]
\centering
\caption{Normalized Eigenvalues $\beta_k = \lambda_k K^*/N$ for $N=20$}\label{tab:eigen}
\renewcommand{\arraystretch}{1.1}
\begin{tabular}{c|c|c|c|c}
\hline
$W$ & $K^*$ & $\bar\beta$ & $\beta_{\max}$ & $\beta_{\min}$ \\
\hline
1 & 3 & 0.97 & 1.23 & 0.60 \\
2 & 5 & 0.97 & 1.35 & 0.64 \\
3 & 7 & 0.97 & 1.50 & 0.67 \\
5 & 11 & 0.97 & 1.78 & 0.67 \\
\hline
\multicolumn{5}{l}{\small $\bar\beta$: average; $\beta_{\max}$, $\beta_{\min}$: extremes over $k=1,\ldots,K^*$.}
\end{tabular}
\end{table}

\vspace{-2mm}
\subsection{Ergodic Capacity Validation}
Fig.~\ref{fig:erg_cap} tests the closed-form ergodic capacity~\eqref{eq:cap_series} against MC for $N=20$, $W\in\{1,2,3,5\}$. Analytical curves agree closely with MC across the full $-5$ to $30$~dB SNR range. At $20$~dB, widening $W$ from $1$ to $5$ adds approximately $2$~bits/s/Hz, confirming both the formula's accuracy and its suitability for fast system-level evaluation. Note that the $2$~bits/s/Hz gain from $W=1$ to $W=5$ requires no extra RF chains or transmit power, in spectral efficiency terms, and it is comparable to adding a second receive antenna in a $1\times2$ maximum-ratio-combining (MRC) system, but achieved with a single chain within a compact aperture.

\begin{figure}[t]
\centering
\includegraphics[width=0.95\columnwidth]{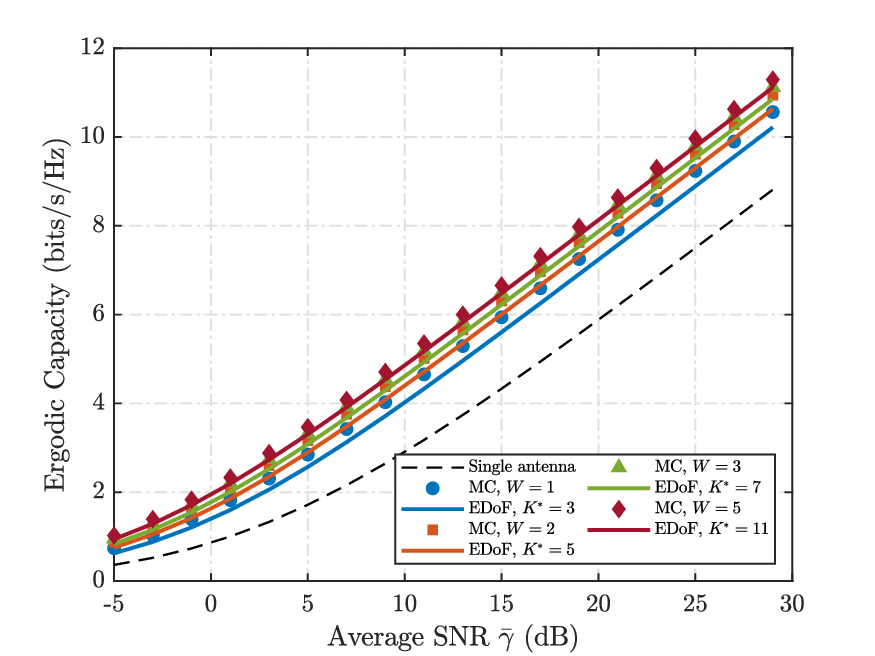}
\caption{Ergodic capacity vs.\ average SNR for $N=20$ and different apertures $W$. Lines: EDoF analytical~\eqref{eq:cap_series}; markers: MC.}\label{fig:erg_cap}
\end{figure}

\vspace{-2mm}
\subsection{Comparison with BCM and VBCM}
Fig.~\ref{fig:edof_bcm} compares EDoF with BCM~\cite{BC24} at $N=40$, $W=3$ ($K^*=7$). The BCM uses $D=4$ equi-correlated blocks of $B=10$ ports with $\rho\approx0.38$. VBCM~\cite{LaiX,TuoW} assigns block-specific parameters and would yield a different curve but shares the same systematic optimistic bias, since both models discard inter-block correlations. Discarding inter-block correlations causes BCM to underestimate the true outage at moderate-to-high SNR, consistent with its known optimistic bias. EDoF and Refined WIM stay above the true outage. The ordering $P_{\mathrm{out}}^{\mathrm{iid}} \leq P_{\mathrm{out}}^{\mathrm{BCM}} \leq P_{\mathrm{out}}^{\mathrm{exact}} \leq P_{\mathrm{out}}^{\mathrm{EDoF}} \leq P_{\mathrm{out}}^{\mathrm{Ref}}$ from Proposition~\ref{prop:ordering} is confirmed. At an outage target of $10^{-2}$, the BCM predicts the required SNR to be roughly $2$~dB less than actually needed; a link budget dimensioned on BCM will therefore fail its reliability specification.

\begin{figure}[t]
\centering
\includegraphics[width=0.95\columnwidth]{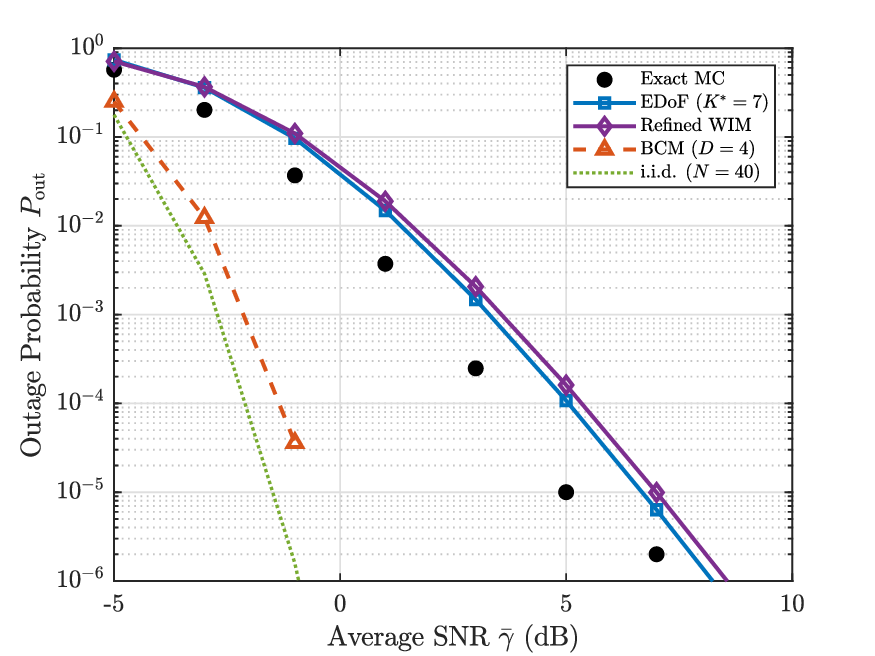}
\caption{EDoF vs.\ BCM outage comparison for $N=40$, $W=3$ ($K^*=7$), $\gamma_{\mathrm{th}}=0$~dB.}\label{fig:edof_bcm}
\end{figure}

\vspace{-2mm}
\subsection{Outage vs.\ Threshold}
In Fig.~\ref{fig:pout_th}, we sweep the outage threshold $\gamma_{\mathrm{th}}$ at $N=20$, $W=3$, $\bar\gamma=10$~dB. Both EDoF and Refined WIM track MC accurately over the full $-10$ to $10$~dB range. The log-scale gap between EDoF and MC stays approximately constant, matching the multiplicative factor $\xi_{K^*}$ of Theorem~\ref{thm:asymptotic}. The single-antenna curve shows the scale of diversity gain available from FAS: at $\gamma_{\mathrm{th}}=0$~dB, the FAS outage is more than two orders of magnitude below the single-antenna value. Since the EDoF-to-MC gap is constant across all threshold values, a single calibration of $10\log_{10}(\xi_{K^*})/K^*\approx0.7$~dB per branch compensates the approximation error at any operating point, making the correction independent of the chosen threshold.

\begin{figure}[t]
\centering
\includegraphics[width=0.95\columnwidth]{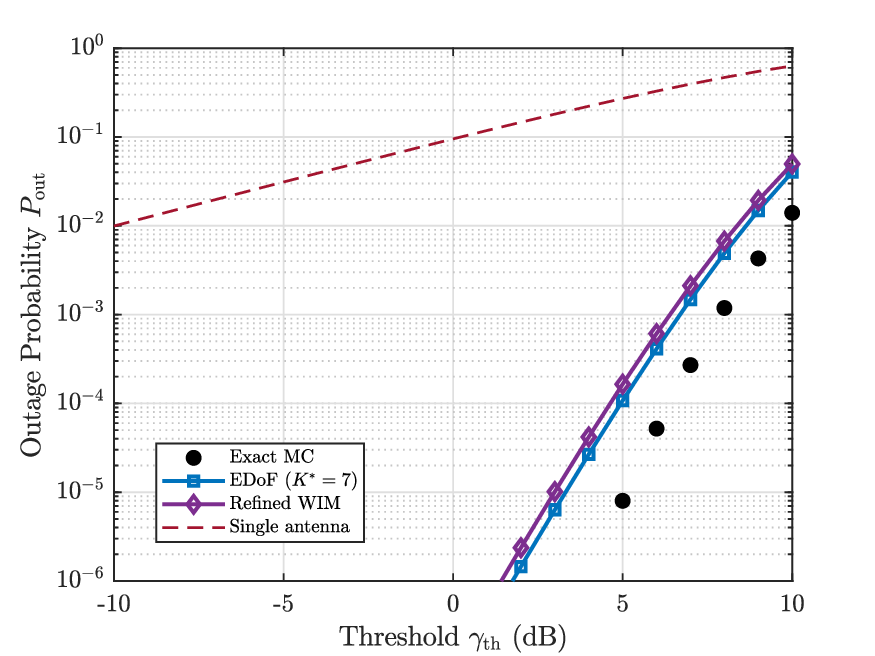}
\caption{Outage probability vs.\ threshold $\gamma_{\mathrm{th}}$ for $N=20$, $W=3$ ($K^*=7$), $\bar\gamma=10$~dB.}\label{fig:pout_th}
\end{figure}

\vspace{-2mm}
\subsection{FAMA Outage Validation}
Fig.~\ref{fig:fama} plots the FAMA outage from Theorem~\ref{thm:fama} versus SNR for $W=3$ ($K^*=7$) and $M \in \{1,2,3,5\}$ users at $\gamma_{\mathrm{th}} = 0$~dB. Several important observations emerge. First, for a single user ($M=1$), the formula reduces exactly to the interference-free EDoF result $(1-e^{-x})^{K^*}$, validating internal consistency. Second, as $M$ increases, the outage rises and develops a characteristic \emph{outage floor} at high SNR, hallmark of interference-limited systems. The floors are at $\approx0.12$ ($M=2$), $\approx0.34$ ($M=3$), and $\approx0.73$ ($M=5$), consistent with the high-SNR limit of~\eqref{eq:fama_pout}. Third, beyond $15$~dB, further SNR increase brings negligible improvement for $M\geq3$, confirming the interference-limited regime. This underscores the importance of aperture-based solutions (larger $W$, hence higher $K^*$) over power-based ones for multi-user FAS~\cite{WongFAMA}.  

\begin{figure}[t]
\centering
\includegraphics[width=0.95\columnwidth]{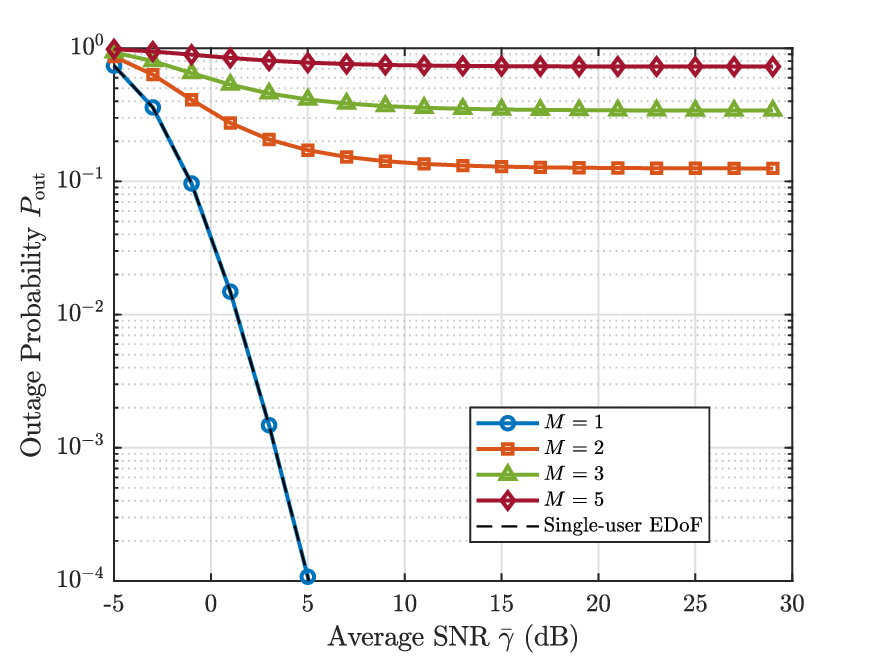}
\caption{FAMA outage probability vs.\ SNR for $W=3$ ($K^*=7$), $\gamma_{\mathrm{th}}=0$~dB, and different numbers of users $M$.}\label{fig:fama}
\end{figure}

\vspace{-2mm}
\subsection{1D vs.\ 2D FAS Comparison}
Fig.~\ref{fig:2d_fas} compares 1D and 2D EDoF outage at $\gamma_{\mathrm{th}}=0$~dB, confirming the multiplicative scaling of Theorem~\ref{thm:2d}. The $2\lambda\times2\lambda$ planar FAS gives $K_{2\mathrm{D}}^*=25$ versus $K^*=5$ for 1D $W=2$; the $3\lambda\times3\lambda$ case gives $K_{2\mathrm{D}}^*=49$ versus $K^*=7$ for 1D $W=3$. At a target outage of $10^{-3}$, the $3\!\times\!3$ planar FAS requires roughly $18$~dB less SNR than a single antenna and about $8$~dB less than the 1D $W=3$ configuration. The $2\!\times\!2$ planar FAS even outperforms 1D $W=3$ by about $3$~dB despite its smaller total aperture. Spreading the aperture across two dimensions is thus more efficient for capturing spatial diversity than extending it in one direction. As both 1D and 2D configurations use a single RF chain, the $8$~dB gain of $3\times3$ over 1D $W=3$ within the same footprint is purely architectural, a direct incentive to adopt planar FAS geometries in compact terminals where antenna area is constrained.

\begin{figure}[t]
\centering
\includegraphics[width=0.95\columnwidth]{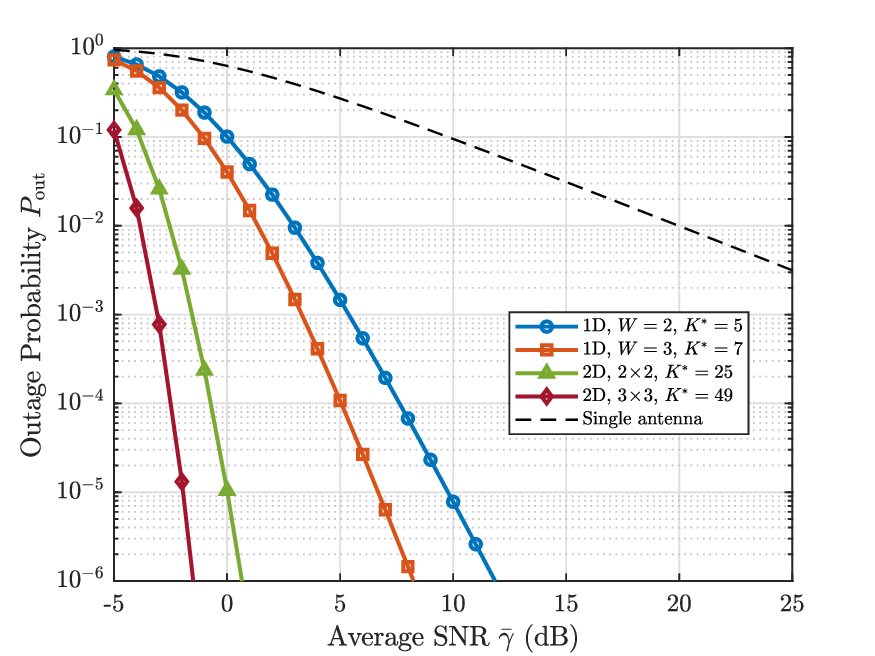}
\caption{EDoF outage comparison: 1D linear FAS vs.\ 2D planar FAS for $\gamma_{\mathrm{th}}=0$~dB.}\label{fig:2d_fas}
\end{figure}

\vspace{-2mm}
\section{Conclusion}\label{sec:conclusion}
This paper revealed that the spatial correlation of a FAS with normalized aperture $W$ is governed by only $K^* = 2\lceil W\rceil+1$ dominant eigenmodes, regardless of how many ports are deployed. This observation directly yields the EDoF outage formula: a simple closed-form expression equivalent to selection combining over $K^*$ independent Rayleigh branches. The formula requires no eigendecomposition and no numerical integration, is provably conservative (it never underestimates the true outage), and gives the exact diversity order $K^*$. Existing block-correlation models, by contrast, systematically underestimate the true outage. A refined WIM approximation was also developed that replaces the equal-power branches of EDoF with eigenvalue-weighted branches $\{\beta_k\}$, tightening the outage bound at moderate SNR while keeping a closed-form product structure. The framework was extended to: closed-form ergodic capacity using order statistics; multi-user FAMA with explicit interference-limited outage floors; and 2D planar FAS where the diversity order $K_{2\mathrm{D}}^* = (2\lceil W_x\rceil+1)(2\lceil W_y\rceil+1)$ grows multiplicatively with the two aperture dimensions. All results were confirmed by Monte Carlo simulations.

\end{document}